\documentclass[conference]{IEEEtran}

\usepackage{times}
\usepackage{algorithmic}
\usepackage{amsmath}
\usepackage{amssymb}
\usepackage{amsthm}
\usepackage{color}
\usepackage{colortbl}
\usepackage{cite}
\usepackage{verbatim}
\usepackage[pdftex]{graphicx}
\usepackage{caption}
\usepackage{subcaption}
\usepackage{url}

\theoremstyle{plain}
\newtheorem{theorem}{Theorem}
\newtheorem{lemma}{Lemma}
\newtheorem{claim}{Claim}

\theoremstyle{definition}
\newtheorem{example}{Example}

\theoremstyle{remark}

\makeatletter
\def\thm@space@setup{\thm@preskip=3pt
\thm@postskip=2pt}
\makeatother

\newtheoremstyle{remboldstyle}
  {}{}{}{}{\bfseries}{.}{.5em}{{\thmname{#1 }}{\thmnumber{#2}}{\thmnote{ (#3)}}}
\theoremstyle{remboldstyle}
\newtheorem{rembold}{Remark}

\newtheoremstyle{defstyle}
  {}{}{}{}{\bfseries}{.}{.5em}{{\thmname{#1 }}{\thmnumber{#2}}{\thmnote{ (#3)}}}
\theoremstyle{defstyle}
\newtheorem{defbold}{Definition}

\newtheoremstyle{remboldstyle}
  {}{}{}{}{\bfseries}{.}{.5em}{{\thmname{#1 }}{\thmnumber{#2}}{\thmnote{ (#3)}}}
\theoremstyle{remboldstyle}
\newtheorem{assumption}{Assumption}

\newcommand{\beq}{\begin{eqnarray}}
\newcommand{\eeq}{\end{eqnarray}}
\newcommand{\n}{\nonumber}

\newcommand{\field}[1]{\mathbb{#1}}

\newcommand{\F}{\field{F}}

\newcommand{\cC}{{\cal C}}

\newfont{\bbb}{msbm10 scaled 500}

\newfont{\bb}{msbm10 scaled 1100}

\newcommand{\ZZ}{\mbox{\bb Z}}

\newcommand{\cv}{{\bf c}}

\newcommand{\mv}{{\bf m}}

\newcommand{\pv}{{\bf p}}

\newcommand{\Gm}{{\bf G}}

\newcommand{\Id}{{\bf I}}

\newcommand{\Lm}{{\bf L}}

\newcommand{\Rm}{{\bf R}}

\newcommand{\Ac}{{\cal A}}
\newcommand{\Bc}{{\cal B}}
\newcommand{\Cc}{{\cal C}}

\newcommand{\Ec}{{\cal E}}

\newcommand{\Kc}{{\cal K}}
\newcommand{\Lc}{{\cal L}}

\newcommand{\Nc}{{\cal N}}

\newcommand{\Pc}{{\cal P}}

\newcommand{\Rc}{{\cal R}}
\newcommand{\Sc}{{\cal S}}

\newcommand{\Xc}{{\cal X}}

\IEEEoverridecommandlockouts

\definecolor{OXO-emph}{RGB}{153,0,0}

\newcommand\ceilb[1]{\left\lceil #1 \right\rceil}

\newcommand{\algrule}[1][.2pt]{\par\vskip.5\baselineskip\hrule height #1\par\vskip.5\baselineskip}
\DeclareMathAlphabet{\mathpzc}{OT1}{pzc}{m}{it}

\begin{document}

\sloppy

\title{Locality and Availability in Distributed Storage}
%
\author{Ankit Singh Rawat, Dimitris S. Papailiopoulos, Alexandros G. Dimakis, and Sriram Vishwanath
\thanks{The authors are with the Dept. of ECE, The University of Texas at Austin, Austin, TX 78751 USA. E-mail: \{ankitsr, dimitris\}@utexas.edu, \{dimakis, sriram\}@austin.utexas.edu.}
\thanks{The authors would like to thank Arya Mazumdar and Natalia Silberstein for valuable discussions.} }



\maketitle


\begin{abstract}
This paper studies the problem of code symbol availability: a code symbol is said to have $(r, t)$-availability if it can be reconstructed from $t$ disjoint groups of other symbols, each of size at most $r$. For example, $3$-replication supports $(1, 2)$-availability as each symbol can be read from its $t= 2$ other (disjoint) replicas, {\em i.e.}, $r=1$. However, the rate of replication must vanish like $\frac{1}{t+1}$ as the availability increases.	

This paper shows that it is possible to construct codes that can support a \textit{scaling number of parallel reads while keeping the rate to be an arbitrarily high constant}. It further shows that this is possible with the minimum distance arbitrarily close to the Singleton bound. This paper {\color{black}also} presents a bound demonstrating a trade-off between minimum distance, availability and locality. Our codes  match  the aforementioned bound and their construction relies on combinatorial objects called resolvable designs.

From a practical standpoint, our codes seem useful for distributed storage applications involving hot data, {\em i.e.}, the information which is frequently accessed by multiple processes in parallel.
\end{abstract}





\section{Introduction}

The simplest way of introducing redundancy in distributed storage systems is $3$-replication, where three replicas of each data block are created. This makes it possible for three parallel reads for each data block. In this paper we introduce a new property that we call \textit{availability} that ensures $t+1$ parallel reads for each data block. We are also concerned with the \textit{locality} $r$ of each read, which measures how many blocks must be read before the desired block can be reconstructed. In this language, $3$-replication allows $t + 1=3$ parallel reads for each block, each with locality $r=1$. However, as we increase the availability $t$ by increasing the replication factor, the rate vanishes like $\frac{1}{t+1}$.

We show that it is possible to construct codes that can support a \textit{scaling number of parallel reads while keeping the rate to be an arbitrarily high constant}. Specifically, one of our constructions 
results in codes of dimension $k$ with availability $t=\Theta(k^{1/3-\epsilon})$ where each read has locality $r=\Theta(k^{1/3})$ for any rate. We further show that this is possible while keeping the minimum distance arbitrarily close to the Singleton bound.

The main motivation for this new property is the application of erasure codes for hot data. 
Current distributed storage systems use various forms of redundancy ranging from block replication to traditional and modern storage codes. It is now well understood that classical codes (such as Reed-Solomon) are highly suboptimal for distributed settings due to the repair problem~\cite{dimakis}. Several 
storage codes have been recently developed, each optimized for a different repair cost metric.
Codes that optimize the number of bits communicated during repairs (a quantity called \textit{repair bandwidth}) were developed, for example, in \cite{dimakis, RSK11,zigzag13, cadambe2011optimal, PapDimCad_hadamard,rashmi2013solution} and references therein. Codes with small disk-I/O were studied in \cite{khan2011search,zigzag13}. Finally, codes that minimize the number of nodes that participate in the repair process, a quantity called {\it locality}, were studied in~\cite{HLM2007,pyramid,Gopalan12, SimpleRegen, PapDim12, PKLK12, SRV12, RKSV12, KPLK12,TamoPapDim13, forbes, cadambe_mazumdar}.
Some of these results have found their way into practice: codes with small locality were recently deployed in Azure production clusters~\cite{azure12}, 
 while others have been tested in Facebook clusters\cite{XORing,rashmi2013solution}.

Code designs with small repair bandwidth and locality are attractive for archival and cold data. 
This is information that is rarely accessed or modified, usually involving back-end systems that store 
massive logs for analytics or backups. It turns out that in these applications there are very large volumes of cold data that must be safely retained.

Another significant family of storage problems involves the management of \textit{hot data}. This is frequently accessed information, often in front-end systems facing end-users. For these applications 
data blocks are frequently accessed, in some cases concurrently by multiple system jobs. To the best of our knowledge, there has been little work on the potential benefits of coding for hot data. A notable exception is
the recent line of work from \cite{shah2013redundant,joshi2013delay,liang2013tofec} and references therein 
that combines queuing theory with coding theory for distributed storage systems. In this paper we explore the 
orthogonal direction of providing multiple ways to reconstruct a single block by reading few other blocks, in parallel. This can be potentially combined with queing theoretic models to analyze performance benefits. 

\noindent {\bf Our Contributions:}  We generalize the definition of locally repairable codes (LRCs): 
an $(n,k,r,t)$-LRC is a systematic code that encodes $k$ information symbols to codewords of length $n$, 
%
so that each of the information symbols has locality $r$, \textit{i.e.} it is a function of $r$ other code symbols. An LRC with {\em all-symbol locality} supports locality $r$ for the parity symbols as well.
In an $(n,k,r,t)$-LRC, for every systematic symbol there exists $t$ disjoint groups, each containing at most $r$ coded symbols that can be used to reconstruct it. We then say that the information symbols have $(r, t)$-availability: {\em i.e.},  each information symbol can be reconstructed by accessing $t$ disjoint groups of other code symbols, each of size at most $r$.

Our first result is a bound on the minimum distance of linear $(n,k,r,t)$-LRCs, where each repair group contains one local parity.
~We then proceed with constructing codes that are optimal with respect to this bound.~
We establish an achievability result, conditional on the existence of certain combinatorial structure for given code 
structures. 
We show that for sub-linear locality and availability $r=\Theta(\frac{\log(k)}{\log\log(k)})$ or $\Theta(k^{1/3})$, and $t =\Theta(r^{1-\epsilon})$, optimal
$(n,k,r,t)$-LRCs exist.
These new codes not only have low locality and high availability, but are arbitrarily high-rate, and have distance asymptotically equal to that of an $(n,k)$ {\color{black}maximum-distance separable} (MDS) code.

{\it Prior work on codes with $(r,t)$-availability}:
In \cite{LlHolOgg2013},  Pamies-Juarez {\em et al.} use projective geometries to construct good codes that enable multiple disjoint repair groups. Asteris {\em et al.} study the availability for repairable fountain codes in~\cite{asteris}.
Tamo {\em et al.} recently constructed codes with good locality and availability properties~\cite{TamoBerg}. 
In a parallel work, Wang {\em et al.} present a more general upper bound on minimum distance of {\em linear codes} with $(r,t)$-availability~\cite{WangZhang}.
They further show the existence of codes that achieve the bound in the asymptotically zero rate regime. 
The tightness of this bound  is open in the general case.
We briefly discuss their bound in Sec.~\ref{sec:bound}, where we extend it to non-linear codes. Batch codes~\cite{batchcodes} also enable parallel reads, in an even stronger sense of allowing multiple reads of different blocks simultaneously. However, current work on Batch codes does not consider fault tolerance. It is interesting to investigate distance properties of Batch codes as we discuss in the conclusions. 

\section{LRCs with $(r,t)$-Availability}
\label{sec:def}
We now formally define $(n,k,r,t)$-LRCs. 
In this paper, we only consider  systematic codes of length $n$, where without loss of generality the first $k$ symbols of a codeword $\cv = (c_1, c_2,\ldots, c_n) $ denote the information symbols. For a positive integer $n$, we use $[n]$ to denote the set $\{1, 2,\ldots, n\}$.

\begin{defbold}
\label{rem:def11}
An $(n, k, r, t)$-LRC  satisfies the following three properties:
\begin{enumerate}
\item For each information (systematic) symbol $c_i$, $i \in [k]$, there exist $t$ subsets $\Gamma_1(i),\ldots, \Gamma_t(i) \subset [n]\backslash \{i\}$, such that $c_i$ is a function of the code symbols indexed by $\Gamma_j(i)$, {\em i.e.}, $\cv_{\Gamma_j(i)}$.
\item $|\Gamma_j(i)| \leq r$, for all $i \in [k]$, $j \in [t]$.
\item $\Gamma_j(i) \cap \Gamma_l(i) =  \emptyset$, for all $i\in[k]$ and $j \neq l \in [t]$.
\end{enumerate}
\end{defbold} 

In the following example, we present a $(7, 3, 2, 2)$-LRC which describes various requirements presented in Definition~\ref{rem:def11}.

\begin{example}
\label{ex:lrc_1}
Consider a systematic code which encodes $3$ infromation symbols $(m_1, m_2, m_3)$ to a codeword 
$\cv = (c_1,\ldots, c_7)$ of length $7$ such that
\begin{align}
 \cv = (m_1, m_2, m_3, m_1, m_1 + m_2, m_2+m_3, m_1+m_3). \nonumber
\end{align}
This code is a $(7, 3, 2, 2)$-LRC as it satisfies the three requirements of Definition~\ref{rem:def11} with 
\begin{align}
\Gamma_1(1) = \{4\}, \Gamma_2(1) = \{2, 5\}, \nonumber \\
\Gamma_1(2) = \{1, 5\}, \Gamma_2(2) = \{3, 6\}, \nonumber \\
\Gamma_1(3) = \{2, 6\}, \Gamma_2(3) = \{1, 7\}. \nonumber
\end{align}
In particular, both $\cv_{\Gamma_1(1)} = c_{4} = m_1$ and $\cv_{\Gamma_2(1)} = (c_{2}, c_{5}) = (m_2, m_1 + m_2)$ can be used to obtain the first information symbol $m_1$. 
\end{example}

It follows from Definition~\ref{rem:def11} that an $(n, k, r, t)$-LRC supports $(r,t)$-availability:
any information symbol can be recovered in parallel by accessing itself and then by accessing code symbols indexed by $t$ disjoint repair groups $\Gamma_1(\cdot),\ldots, \Gamma_t(\cdot)$ associated with it. In terms of locality, Definition~\ref{rem:def11} ensures locality $r$ for the information symbols, {\em i.e.}, {\em information-symbol locality}. If an $(n, k, r, t)$-LRC allows for locality of $r$ for all $n$ coded symbols, {\em i.e.}, {\em all-symbol locality}, then it is referred to as an $(n, k, r, t)$-LRC with all-symbol locality. Throughout this paper we consider codes with $(r,t)$-availability {\it only} for the information symbols.

In the next two sections, we establish new distance bounds for  $(n,k,r,t)$-LRCs and then proceed with presenting optimal code constructions based on certain combinatorial structures.

Before we proceed, let us introduce some notation. Let $m$ be the total number of distinct subsets (local groups) of the type $\Gamma_j(i) \cup \{i\}$, according to Definition~\ref{rem:def11}. We use a $k \times m$ {\it membership} matrix of $0$s and $1$s, call it $\mathbf{R}$, to denote the information symbols participating in these $m$ subsets of $[n]$:
a $1$ in the entry $(i,j)$ of ${\bf R}$ means that $i \in [k]$ participates in the $j$-th local group. Hence, each row of ${\bf R}$ indexes a systematic symbol and each column a local group. For the code described in Example~\ref{ex:lrc_1}, we have $4$ distinct subsets (local groups) of the form $\Gamma_j(i) \cup \{i\}$:
\begin{align}
\{1, 4\}, \{1, 2, 5\}, \{2, 3, 6\}, \{1, 3, 7\}.
\end{align}
These correspond to the $3 \times 4$ matrix
\begin{align}
\Rm = \left(\begin{array}{cccc}
1 & 1 &0 & 1 \\
0 & 1 & 1 & 0 \\
0 & 0 &1 &1 \end{array} \right). \n
\end{align}

\section{Upper Bound on Minimum Distance of $(n, k, r, t)$-LRCs}
\label{sec:bound}

Here, we present upper bound on the distance of linear $(n, k, r, t)$-LRCs.
In our main theorem, we assume a simple condition for the codes:
each repair group ${\Gamma_j(i)}$ contains {\it only} 1 parity symbol.
This condition is later lifted, and a more general bound is presented.
Before stating our results, we present a definition of the minimum distance of a code.

\begin{defbold}
\label{def:dmin}	
The minimum distance of a code $\mathcal{C}$ is equal to 
\begin{equation}
d_{\min}(\Cc) =n-|\Sc^*|,
\label{eq:dmin_def}
\end{equation}
where  $\mathcal{S}^{*} \subset [n]$ denotes a maximum cardinality set such that the encoded symbols indexed by $\Sc^{\ast}$ are not sufficient to reconstruct all $k$ information symbols.
\end{defbold}

\begin{lemma}
\label{lem:m_bound}
Let $\mathbf{R}$ be a $k \times m$ matrix with entries in $\{0,1\}$, as defined in Sec.~\ref{sec:def}. 
Then, the number of columns in ${\bf R}$ satisfies the inequality $$m \geq \ceilb{\frac{kt}{r}}.$$
\end{lemma}
\begin{proof}
Note that, 
\begin{align}
\label{eq:bound_m1}
\text{$\#$ of $1$s in $\mathbf{R}$} &= \sum_{i = 1}^{m}\text{$\#$ of $1$s in $\mathbf{R}(:,i)$} \leq mr, \\
\text{$\#$ of $1$s in $\mathbf{R}$} &= \sum_{i = 1}^{k}\text{$\#$ of $1$s in $\mathbf{R}(i,:)$} \geq kt. \label{eq:bound_m2}
\end{align}
Here, $\mathbf{R}(:,i)$ and $\mathbf{R}(i,:)$ denote $i$-th column and $i$-th row of $\mathbf{R}$, respectively. 
The first inequality is because each of $m$ local groups, {\em i.e.}, a subset of the form $\Gamma_j(i)\cup \{i\}$, contain at most $r$ elements from $[k]$. The second is due to the fact that for each information symbol there are $t$ disjoint repair groups. Therefore, each $i \in [k]$ appears in at least $t$ local groups. Using \eqref{eq:bound_m1} and \eqref{eq:bound_m2}, we get that $mr \geq  kt$, or 
\begin{align}
&~~~m \geq \ceilb{\frac{kt}{r}}.
\end{align}
\end{proof}

\noindent  We use the above lemma to obtain the following theorem.
\begin{theorem}
\label{prop:bound_restricted}
Let $\cC$ be a linear $(n, k, r, t)$-LRC such that any repair group defined by $\mathbf{R}$ contains only $1$ parity symbol. 
Then, the distance of the code is bounded as
\begin{align}
\label{eq:dmin_bound}
d_{\min}(\Cc) \leq  n - k - \ceilb{\frac{kt}{r}} + t + 1.
\end{align}
\end{theorem}
\begin{proof}
Given the assumption that each repair group associated with $\mathbf{R}$ has $1$ parity, we have at least $m$ local parities in our code $\Cc$, one for each column of $\mathbf{R}$. 
Keeping Definition~\ref{def:dmin} in mind, we now construct a set $\Sc \subset [n]$ such that one can not reconstruct all $k$ information symbols from the encoded symbols indexed by $\Sc$.
 We consider two cases:

\textit{Case $1$:}~There is an information symbol, say $i$, which has exactly $t$ disjoint repair groups associated with it, {\em i.e.}, the $i$-th row of $\mathbf{R}$ has exactly $t$ ones. Consider the set $\Sc  = \left([k]\backslash{i}\right) \cup \Pc_{\mathbf{R}_{i}}$, where $\Pc_{\mathbf{R}_{i}}$ denotes the set of local parities associated with those columns of $\mathbf{R}$ that have zero as their $i$-th entry. By the choice of $i$, we have $|\Pc_{\mathbf{R}_{i}}| = m - t$. Note that we can not recover the $i$-th information symbol from the encoded symbols indexed by the set $\Sc$. This implies that $|\Sc^{\ast}| \geq |\Sc| = k - 1 + m - t$. Therefore, it follows from \eqref{eq:dmin_def} that
\begin{align}
\label{eq:dmin_1}
d_{\min}(\Cc) \leq n - |\Sc| = n - k - m + t + 1.
\end{align} 
Combining Lemma~\ref{lem:m_bound} and \eqref{eq:dmin_1}, we get
\begin{align}
\label{eq:dmin_bound_case1}
d_{\min}(\Cc) \leq  n - k - \ceilb{\frac{kt}{r}} + t + 1.
\end{align}

\textit{Case $2$:} The row with minimum number of $1$s in $\mathbf{R}$ has weight $t' > t$. In this case, we have $rt' \leq \text{$\#$ of $1$s in $\mathbf{R}$} \leq mr$, which gives us 
\begin{align}
\label{eq:new_m}
m \geq \ceilb{\frac{kt'}{r}}.
\end{align}

Let us assume that $i \in [k]$ is such that the $i$-th row in $\mathbf{R}$ has exactly $t'$ ones. 
Similar to case $1$, consider $\Sc  = \left([k]\backslash{i}\right) \cup \Pc_{\mathbf{R}_{i}}$. Note that the $i$-th information symbol can not be recovered from the encoded symbols indexed by $\Sc$ and $|\Sc^{\ast}| \geq |\Sc| = k - 1 + m - t'$. Using \eqref{eq:dmin_def} and \eqref{eq:new_m}, 
\begin{align}
\label{eq:dmin_bound_case2}
d_{\min}(\Cc) \leq  n - k - \ceilb{\frac{kt'}{r}} + t' + 1.
\end{align}

Note that for $r \leq k$, the right-hand side of \eqref{eq:dmin_bound_case1} is greater than that of \eqref{eq:dmin_bound_case2}. Therefore, we can combine the two cases to obtain the bound in \eqref{eq:dmin_bound}.
\end{proof}


\begin{rembold}
Note that for $t = 1$, the bound in \eqref{eq:dmin_bound}  reduces to the distance bound for codes with $r$-locality~\cite{Gopalan12, PapDim12}, {\em i.e.},
\begin{align}
d_{\min}(\Cc) \leq  n - k - \ceilb{\frac{k}{r}} + 2. \nonumber
\end{align}
\end{rembold}

Next, we present a more general bound on the minimum distance of an $(n, k, r, t)$-LRC. This bound does not assume linearity of a code or that the repair groups associated with $\mathbf{R}$ have exactly one (local) parity. 
\begin{theorem}
\label{prop:general_bound}
For an $(n, k, r, t)$-LRC, linear, or non-linear, we have  
\begin{align}
\label{eq:general_bound}
d_{\rm min}(\cC) \leq n - k - \ceilb{\frac{t(k-1) + 1}{t(r-1) + 1}} + 2
\end{align}

\end{theorem}

\begin{proof}
See Appendix~\ref{appen:general_bound} for the proof.
\end{proof}

\begin{rembold}
Wang {\em et al.} established \eqref{eq:general_bound} for linear codes and show the existence of linear codes that attain the bound when $n \geq k(rt + 1)$~\cite{WangZhang}. However, the tightness of \eqref{eq:general_bound} remains an open question in the general case and for codes with high rate.
\end{rembold}

\section{Achievability results for $(n, k, r, t)$-LRCs}
\label{sec:achievability_thms}
In this section, we present explicit constructions for codes with $(r, t)$-availability, and analyze their minimum distance. In particular, in Sec.~\ref{sec:local_information} we design $(n, k, r, t)$-LRCs  by modifying {\color{black}the} Pyramid code construction of \cite{pyramid}. {\color{black}Then, in} Sec.~\ref{sec:all_symbol} {\color{black}we use} Gabidulin codes in order to obtain $(n, k, r, t)$-LRCs with all-symbol locality. 
{\color{black}Finally,} Sec.~\ref{subsec:finding_S} describes the role of resolvable designs in {\color{black}our} proposed constructions. Note that the resolvable designs have been previously used to construct codes in other settings, {\em e.g.} see \cite{farzad}.

Before describing {\color{black}the} code constructions, we present in Sec.~\ref{sec:resolvable} and Sec.~\ref{subsec:gabidulin}  {\color{black}a} brief introduction to resolvable designs and Gabidulin codes. 

\subsection{Background on resolvable designs}
\label{sec:resolvable}
 
Here, we briefly introduce resolvable designs. Readers may refer to \cite{IonShri} for a detailed treatment of this subject. 

\begin{defbold}
\label{def:2-design}
A pair $(\Xc, \Bc)$, where $\Xc = (x_1, \ldots, x_k)$ is a $k$-element set and $\Bc = (B_1, \ldots, B_b)$ is a family of subsets (blocks) of $\Xc$, is called a $2$-$(k, b, c,  r, \lambda)$-resolvable design if it satisfies the following properties: (i) $|B_{i}| = r$~for all~$i \in [b]$, (ii) every pair $(x, y) \subset \Xc$ is present in exactly $\lambda$ blocks (subsets) in $\Bc$, and (iii) $\Bc$ comprises $c$ disjoint collections of blocks (namely parallel classes) $\Ec_1,\ldots, \Ec_c \subset \Bc$ such that $|\{B \in \Ec_{i} : x_{j} \in B\}| = 1$, for all $i \in [c]$ and $j \in [v]$, {\em i.e.}, blocks in each parallel class partition the set $\Xc$. 
\end{defbold}

{\color{black}We use a} $k \times b$ matrix $\Id_{(\Xc, \Bc)}$ with $0$s and $1$s {\color{black}to denote} the incidence matrix of a design $(\Xc, \Bc)$, {\color{black}where} $\Id_{(\Xc, \Bc)}(i,j) = 1$ {\color{black}if only if}  $x_i \in B_j$.  

In this paper, we focus on $2$-resolvable designs with $\lambda = 1$. Note that $\lambda = 1$ enforces that $x \in \Xc$ is the only common element in any $2$ blocks containing $x$. We now present an example of $2$-($k, b, c,  r, \lambda = 1$)-resolvable design which is obtained as a solution to Kirkman's schoolgirl problem\footnote{http://en.wikipedia.org/wiki/Kirkman's\underline{ }schoolgirl\underline{ }problem}.

\begin{example}
Let $\Xc = \{1, 2, \ldots, 15\}$. Fig.~\ref{fig:kirkman} represents a $2$-$(15, 35, 7, 3, 1)$-resolvable design over elements of the set $\Xc$. 
The {\color{black}$b=35$ different $3$-element sets} in Fig.~\ref{fig:kirkman} denote {\color{black}the} blocks in the design. Note that any $2$ elements appear together in exactly $\lambda = 1$ block. The blocks from each of the $7$ parallel classes $\Ec_1,\ldots, \Ec_7$ partition the set $\Xc$.
\end{example}
\begin{figure}[t!]
\begin{center}
\includegraphics[width=0.34\textwidth]{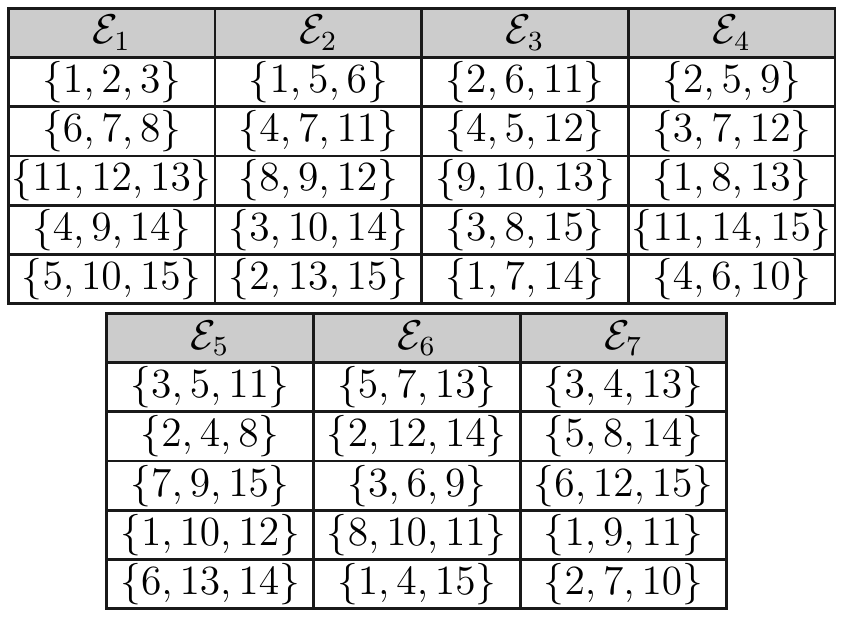}
\end{center}
\caption{An example of $2$-$(15, 35, 7, 3, 1)$ resolvable design.}
\label{fig:kirkman}
\end{figure}

\subsection{Gabidulin codes}
\label{subsec:gabidulin}
Gabidulin codes are an example of maximum rank distance (MRD) codes~\cite{Gab85}.
{\color{black}Gabidulin codes are MDS}. Encoding $\mathbf{m} = (m_1, m_2,\ldots, m_{\mathcal{K}}) \in \F_{q^M}^{\Kc}$ to a codeword $\cv$ of  an $[\mathcal{N}, \mathcal{K}, \mathcal{D}]$ Gabidulin code over $\mathbb{F}_{q^M}$ is performed by evaluating {\color{black}a} data polynomial $\mathpzc{f}(y) = \sum_{i = 1}^{\mathcal{K}}m_{i}y^{q^{i-1}}$ at $\Nc$-linearly independent (over $\mathbb{F}_{q}$) points in $\mathbb{F}_{q^M}$, say $\{y_1, y_2,\ldots, y_{\Nc}\}~\subset \mathbb{F}_{q^M}$, {\em i.e.}, $\mathbf{c}= (\mathpzc{f}(y_1),\ldots, \mathpzc{f}(y_{\Nc}))$. Note that the above encoding process can be represented as $\cv = \mv\Gm_{\rm Gab}$, where
\begin{align}
\Gm_{\rm Gab} = [\Gm^{1}_{\rm Gab}~|~ \Gm^{2}_{\rm Gab}] =  \left(\begin{array}{cccc}
y_{1} & y_{2} &\cdots & y_{\mathcal{N}} \\
y^q_{1} & y^q_{2} & \cdots & y^q_{\mathcal{N}} \\
\vdots & \vdots & \ddots & \vdots \\
y^{q^{\mathcal{K}-1}}_1 & y^{q^{\mathcal{K}-1}}_2 & \cdots &y^{q^{\mathcal{K}-1}}_{\mathcal{N}} \end{array} \right). \n 
\end{align}
Here, $\Gm^{1}_{\rm Gab}$ and $\Gm^{2}_{\rm Gab}$ denote the first $\mathcal{K}$ and the last $\mathcal{N} - \mathcal{K}$ columns of $\Gm_{\rm Gab}$, respectively.


\subsection{Construction of $(n, k, r, t)$-LRCs}
\label{sec:local_information}

In this subsection we present a construction for $(n, k, r , t)$-LRCs when $r|k$ and the following assumption holds. 
\begin{assumption}
\label{assump:1}
There is a $k \times t\frac{k}{r}$ matrix $\Rm$ with $0$s and $1$s such that (i) each column of $\Rm$ has $r$ nonzero entries with supports of the columns of $\Rm$ giving $t$ partitions of $[k]$, and (ii) {\color{black}the} supports of any two rows of $\Rm$ intersect at at most $1$ position.
\end{assumption}
Let $\mathbf{R}_1,\ldots, \mathbf{R}_t$ denote the $t$ collection of $\frac{k}{r}$ columns of $\mathbf{R}$ whose supports constitute $t$ distinct partitions of the set $[k]$.  
\begin{figure}[htbp]
	\centering
		\includegraphics[width=0.48\textwidth]{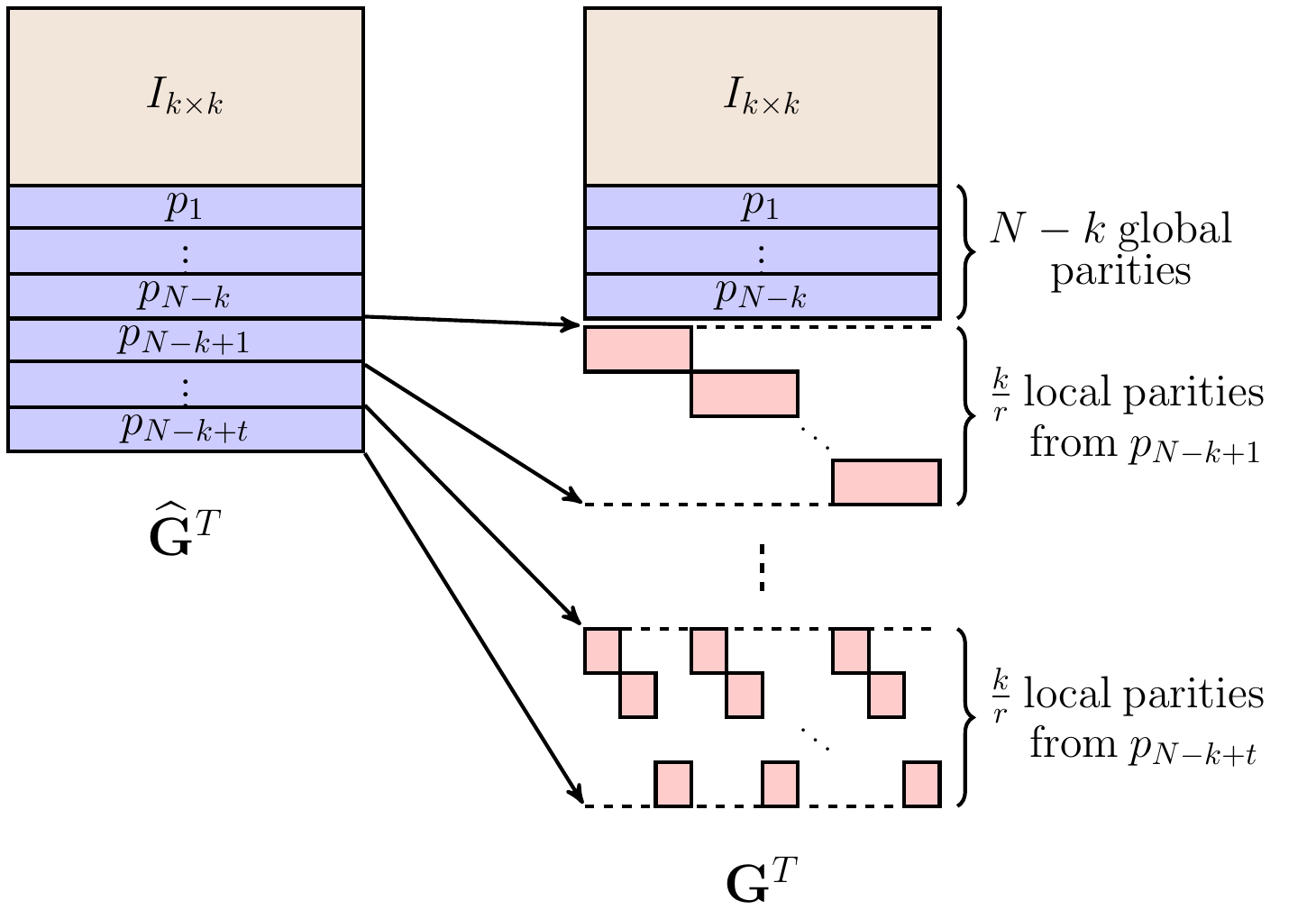}
           \caption{Illustration of the construction of a generator matrix $\Gm$ for an $\left(n = N + t\frac{k}{r}, k, r, t\right)$-LRC from generator matrix $\widehat{\Gm}$ of a systematic $(N+t, k)$ MDS code. First $N$ columns of $\Gm$ are exactly the same as first $N$ columns of $\widehat{\Gm}$. Last $t\frac{k}{r}$ columns of $\Gm$ are obtained by splitting each of last $t$ columns of $\widehat{\Gm}$ into $\frac{k}{r}$ columns of support $r$ each.}
	\label{fig:Pyramid_availability}
\end{figure}
Given a generator matrix $\widehat{\Gm}$ of a systematic $(N + t, k)$ MDS code, a generator matrix $\Gm$ for an $(n = N + \frac{kt}{r}, k, r, t)$-LRC is constructed as follows:

\begin{figure*}
        \centering
        \begin{subfigure}[b]{0.48\textwidth}
                \centering
                \includegraphics[width=0.97\textwidth]{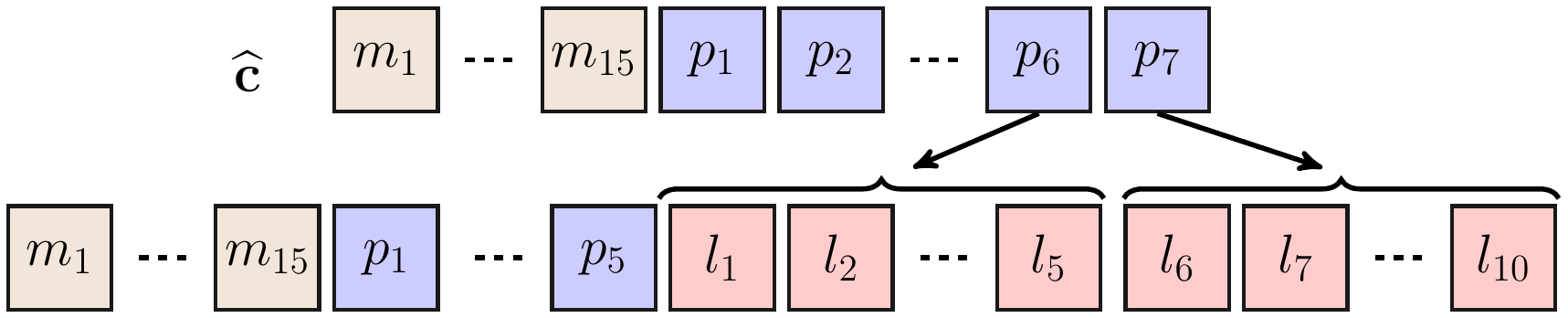}
                \caption{A $(30, 15, 3, 2)$-LRC}
                \label{fig:example1}
        \end{subfigure}%
       ~
        \begin{subfigure}[b]{0.48\textwidth}
   		\footnotesize
    		\centering
   		 \begin{tabular}{| l | l |}
   		 \hline
    		$l_1$ = $a_1m_1$ + $a_2m_2$ + $a_3m_3$ & $l_6$ = $b_1m_1$ + $b_5m_5$ + $b_6m_6$ \\ \hline
    		$l_2$ = $a_6m_6$ + $a_7m_7$  + $a_8m_8$ & $l_7$ = $b_4m_4$ + $b_7m_7$ + $b_{11}m_{11}$ \\ \hline
    		$l_3$ = $a_{11}m_{11}$ + $a_{12}m_{12}$ + $a_{13}b_{13}$ & $l_8$ = $b_{8}m_{8}$ + $b_{9}m_{9}$ + $b_{12}m_{12}$ \\ \hline
    		$l_4$ = $a_4m_4 $+ $a_9m_9$ +  $a_{14}m_{14} $ & $l_9$ = $b_{3}m_{3}$ + $b_{10}m_{10}$ + $b_{14}m_{14}$ \\ \hline
    		$l_5$ = $a_5m_5$ + $a_{10}m_{10}$ + $a_{15}m_{15}$ & $l_{10}$ = $b_{2}m_{2}$ + $b_{13}m_{13}$ + $b_{15}m_{15}$  \\
    		\hline 
   		\end{tabular} 
		\caption{Local parities in Example~\ref{example1}.}
		\label{tab:example1}
        \end{subfigure}
        \caption{A $(30, 15, 3, 2)$-LRC $\Cc$ obtained by Construction I, as described in Example~\ref{example1}. 
Fig.~\ref{tab:example1} illustrates the local parities  $(l_1,\ldots, l_{10})$ of $\Cc$ which are obtained by breaking two parities $p_6 = \sum_{i = 1}^{15}a_im_i$ and $p_7 = \sum_{i = 1}^{15}b_im_i$ of $\widehat{\Cc}$.}
\end{figure*}

\subsubsection{Construction I}
\label{subsec:construction_I}
\begin{itemize}
\item {\color{black}The} first $N$ columns of ${\Gm}$  {\color{black}are the} first $N$ columns of  $\widehat\Gm$. (See Fig.~\ref{fig:Pyramid_availability}).
\item For each $i \in [t]$, split   {\color{black}the} $(N+i)$-th column of $\widehat{\Gm}$ into $\frac{k}{r}$ columns of weight $r$ each,   {\color{black}such that} their supports   {\color{black}are according} to the $\frac{k}{r}$ columns in $\Rm_i$.
\end{itemize}
Note that   {\color{black}the} first $k$ columns of $\Gm$ correspond to systematic symbols.
  {\color{black}Then, the} columns   of $\Gm$ {\color{black}from} $k+1$ to $N$  are associated with global parities. 
  {\color{black}Finally,} the last $t\frac{k}{r}$ columns of $\Gm$, obtained by splitting {\color{black}the} $t$ last columns of $\widehat{\Gm}$, correspond to local parities.  

{\color{black}Note that Construction I} differs from the original Pyramid codes~\cite{pyramid}, {\color{black}as} the {\color{black}non-overlapping support} requirement on  the local parities is not present in \cite{pyramid}. 
 
\begin{rembold}
\label{rem:design_use}
Resolvable designs provide one way to obtain the matrix $\Rm$ utilized in Construction I. 
Given a $2$-$(k, b, c,  r, 1)$-resolvable design $(\Xc, \Bc)$ with $c \geq t$ parallel classes, one can take $\Rm$ to be the sub-matrix of the incidence matrix $\Id_{(\Xc, \Bc)}$ associated with first $t$ parallel classes. Note that the columns of $\Id_{(\Xc, \Bc)}$ correspond to blocks in $\Bc$, and {\color{black}the} support of {\color{black}the} $j$-th column indicates the elements of $\Xc$ that appear in block $B_j \in \Bc$. (See Sec.~\ref{sec:resolvable}.) For $i \in [t]$, $\Rm_i$ denotes {\color{black}the} $\frac{k}{r}$ columns of $\Rm$ associated with {\color{black}the} $i$-th parallel class in $(\Xc, \Bc)$. 
\end{rembold}

See Sec.~\ref{subsec:finding_S} for further discussion on finding the matrix $\Rm$. Next, we present an example to describe Construction I. 

\begin{example}
\label{example1}
Here, we present an $(30, 15, 3, 2)$-LRC which is obtained by Construction I. 
As described in Construction I, we take a systematic $(N+t, k) = (22, 15)$ MDS code with generator matrix $\widehat{\Gm}$. Let $\mv = (m_1, \ldots, m_{15})$ and $\widehat{\cv} = (m_1,\ldots, m_{15}, p_{1},\ldots, p_{7}) = \mv\widehat{\Gm}$. (See Fig.~\ref{fig:example1}.) The codeword $\cv$ for message $\mv$ in {\color{black}our} $(30, 15, 3, 2)$-LRC is obtained by splitting {\color{black}the} $t = 2$ columns in $\widehat{\Gm}$ which corresponds to parities $p_{6}$ and $p_{7}$ in $\widehat{\cv}$. Let $(l_1,\ldots, l_{10})$ denote the $t\frac{k}{r} = 2 \times \frac{15}{3} = 10$ local parities obtained in this manner. 

 Assuming that $p_6 = \sum_{i = 1}^{15}a_im_i$ and $p_7 = \sum_{i = 1}^{15}b_im_i$, Fig.~\ref{tab:example1} describes the local parities $\{l_j\}_{j = 1}^{10}$. We use {\color{black}the} first $2$ parallel classes of the design from  Fig.~\ref{fig:kirkman} to generate {\color{black}the} local parities (see Remark~\ref{rem:design_use}). In particular, each parallel class of the design in Fig.~\ref{fig:kirkman} gives a {\color{black}partition} of $\{m_1,\ldots, m_{15}\}$ into $\frac{k}{r} = 5$ sets of size $r = 3$. Each of these $5$ sets corresponds to $1$ local parity as evident from Fig.~\ref{tab:example1}. Note that $\Cc$ has $(3, 2)$-availability. For example, $m_1$ can be reconstructed by $\{m_2, m_3, l_1\}$ and $\{m_5, m_6, l_6\}$. Similarly, $\{m_5, m_{10}, l_5\}$ and $\{m_{2}, m_{13}, l_{10}\}$ allow us to recover $m_{15}$.
\end{example}

The following result establishes that Construction I generates $(n, k, r, t)$-LRCs which attain the distance bound in \eqref{eq:dmin_bound}. 
%
\begin{theorem}
\label{thm:Pyramid}
Let $r | k$ and {\color{black}let} Assumption~\ref{assump:1} hold. Then, Construction I gives an $(n = N + t\frac{k}{r}, k, r, t)$-LRC with 
\begin{align}
\label{eq:Pyramid_dmin_proof}
d_{\min}(\Cc)  = n - k - \ceilb{\frac{kt}{r}} + t + 1.
\end{align}
\end{theorem}
\begin{proof}
We use columns of $\Rm$ to construct $t\frac{k}{r}$ local parities from $t$ parities of an MDS code in Construction I. The requirement on columns of $\Rm$ to produce $t$ partitions of $[k]$ ensures that each systematic symbol is covered by $t$ local parities of weight $r$. The restriction on the size of intersection of support of any two rows of $\Rm$ translates to the fact that any two of the $t$ repair groups for a systematic symbol consists of disjoint code symbols. 

In order to establish \eqref{eq:Pyramid_dmin_proof}, we show that an $(n, k)$ code obtained from construction I can correct any pattern of $n - k - \ceilb{\frac{kt}{r}} + t =  N + t\frac{k}{r} - k  - \frac{kt}{r} + t = N - k + t$ node erasures. The proof here essentially follows the arguments presented in \cite{pyramid}. 

We index symbols of a codeword in $\mathcal{C}$ from $1$ to $n$. Let $\mathcal{I}, \mathcal{P}^{\rm gbl}, \mathcal{P}^1$ denote the  sets of indices of systematic symbols, global parities, and local parities introduced in Construction I (for $(r, t)$-availability), respectively.  Let $\{\mathcal{P}_i^1\}_{i = 1}^{t}$ denote the sets of indices of local parities obtained from $(N + i)$-th column  of $\widehat{\Gm}$. Note that $|\mathcal{P}^1_i| =\frac{k}{r}$~for all $i \in [t]$ and $\mathcal{P}^1 = \bigcup_{i = 1}^{t}\mathcal{P}_{i}^1$. 

Next, we consider two cases for node erasure patterns:

\textit{Case 1:} There are at most $ N  - k$ erasures among the symbols indexed by the set $\mathcal{L} = \mathcal{I}\cup\mathcal{P}^{\rm gbl}$. Note that the code obtained from puncturing $\mathcal{C}$ on $[n]\backslash \mathcal{L}$, {\em i.e.}, $\mathcal{C}_{\mathcal{L}}$, is an $(N, k)$ MDS code. Therefore, message symbols $\mathbf{m}$ can be recovered from $\mathcal{C}_{\mathcal{L}}$ even after any pattern of at most $ N  - k$ erasures in the symbols indexed by the set $\mathcal{L}$.

\textit{Case 2:} There are $ N - k + x$~($0 < x  \leq t$) erasures in $\Cc_{\Lm}$ and $t - x$ erasures among the symbols indexed by the set $\mathcal{P}^1$. In this case, we obtain $k - x$ symbols of a codeword in an $(N + t, k)$ MDS code with generator matrix $\widehat{\Gm}$ from unerased symbols of $\Cc_{\mathcal{L}}$.

Note that there are $t  - x$ erasures in $\Cc_{\mathcal{P}^1}$. In the worst case, these erasure are spread in $t - x$ sets in $\{\mathcal{P}_{1}^1, \ldots, \mathcal{P}_{t}^{1}\}$. Let $\{\mathcal{P}_{i_1}^{1}, \ldots, \mathcal{P}_{i_{x}}^{1}\}$ be the sets corresponding to local parities that do not have any erasures. We can combine $\frac{k}{r}$ local parities associated with each of these sets to obtain $x$ global parities of  the $(N + t, k)$ MDS code with generator matrix $\widehat{\Gm}$. Combining these with the symbols obtained from $\Cc_{\Lc}$, we have $k$ symbols of a codeword in the MDS code with generator matrix $\widehat{\Gm}$. This allows us to recover ${\mathbf{m}}$. 

This completes the proof.
\end{proof}

\subsection{Construction of $(n, k, r, t)$-LRCs with all-symbol locality}
\label{sec:all_symbol}


We now utilize Gabidulin codes~\cite{Gab85} to  construct $(n, k, r, t)$-LRCs with all-symbol locality when $r|k$, $r|N$, and Assumption~\ref{assump:1} holds. Gabidulin codes have been previously used to obtain codes with all-symbol locality in \cite{SRV12, RKSV12}.

\subsubsection{Construction II}
\label{subsec:construction_II}
Let $\Gm_{\rm Gab}$ be a generator matrix of an $[\Nc$ = $N + t -1, \Kc$ = $k]$ Gabidulin code. 
We transform $\Gm_{\rm Gab}$ into $\bar{\Gm}_{\rm Gab}$ = $\left(\Gm^{1}_{\rm Gab}\right)^{-1}\Gm_{\rm Gab}$ = $[\Id~|~\left(\Gm^{1}_{\rm Gab}\right)^{-1}\Gm^{2}_{\rm Gab}]$, which is a generator matrix of a systematic $(N + t - 1, k)$ MDS code. See Sec.~\ref{subsec:gabidulin} for the definitions of $\Gm^{1}_{\rm Gab}$ and $\Gm^{2}_{\rm Gab}$. 

Given the matrix $\bar{\Gm}_{\rm Gab}$, {\color{black}we construct} an $(n, k, r, t)$-LRC with all-symbol locality as follows:

\begin{figure}[htbp]
	\centering
		\includegraphics[width=0.38\textwidth]{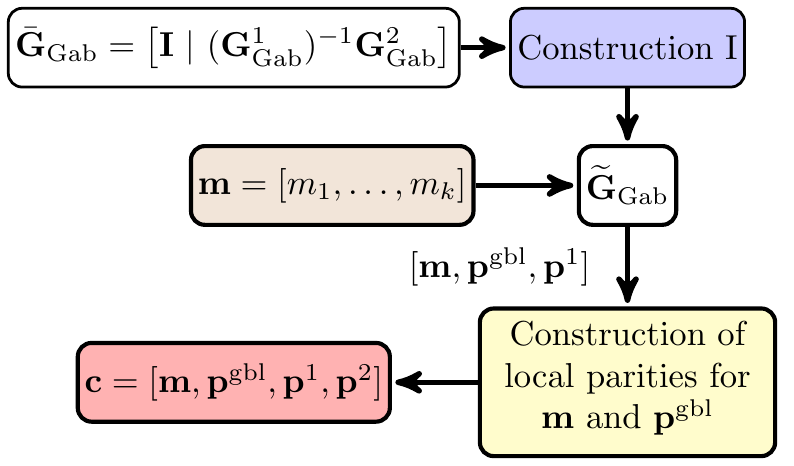}
          \caption{Description of Construction~II. $\mathbf{p}^{\rm gbl}$ and $\mathbf{p}^{1}$ denote global and local parities added in step $1$. $\pv^{2}$ represent the local parities added in step $2$ for all-symbol locality.}
	\label{fig:all_symbol_code1}
\end{figure}

\textbf{Step $1$:} Construct an $(\widetilde{n}$ = $N + (t  - 1)\frac{k}{r}, k, r, t-1)$-LRC from $\bar{\Gm}_{\rm Gab}$ using Construction I with $\Rm_1,\ldots, \Rm_{t-1}$. Let $\widetilde{\Gm}_{\rm Gab}$ denote the generator matrix of the obtained code (see Fig.~\ref{fig:all_symbol_code1}.)

\textbf{Step $2$:} Partition the systematic symbols and global parities of the codeword obtained in the previous step into $\frac{k}{r}  + \frac{N-k}{r} = \frac{N}{r}$ groups, {\color{black}each} of size  $r$. Then introduce $\frac{N}{r}$ local parities for each of these groups ($\frac{k}{r}$ for systematic symbols and $\frac{N-k}{r}$ for global parities). The coefficients of these local parities are chosen to be nonzero elements from the base field $\F_q$. In this step, we partition {\color{black}the} systematic symbols according to the supports of the columns in $\Rm_t$. 
 In this way, we obtain an $(n = N + \frac{N}{r} + (t  - 1)\frac{k}{r}, k, r, t)$-LRC with all-symbol locality. 
\begin{rembold}
The requirement of $r|N$ can be relaxed by following the ideas from \cite{RKSV12}. For the ease of exposition we only consider the case $r \mid N$ here.
\end{rembold}

\begin{theorem}
\label{thm:all_symbol}
Let $r | k$, $r | N$, and {\color{black}let} Assumption~\ref{assump:1} hold. Then, Construction II gives an $(n = N + \frac{N}{r} + (t  - 1)\frac{k}{r}, k, r, t)$-LRC with all-symbol locality and 
\begin{align}
\label{eq:all_symbol_dmin_proof}
d_{\min}(\mathcal{C}) = n - k - \ceilb{\frac{kt}{r}} + t + 1.
\end{align}
\end{theorem}
\begin{proof}
The proof of all-symbol locality and $(r, t)$-availability for $\mathcal{C}$ follows from the structure of $\Rm$ used in the Construction II. Here, we show that $\mathcal{C}$ allows original message symbols to be recovered even after any pattern of  $n - k - \ceilb{\frac{kt}{r}} + t  = N + \frac{N}{r} + (t  - 1)\frac{k}{r} -  k - \ceilb{\frac{kt}{r}} + t = N + \frac{N}{r} - k - \frac{k}{r} + t$ node erasures. This along with \eqref{eq:dmin_bound} give us the result in \eqref{eq:all_symbol_dmin_proof}.

We index symbols of a codeword in $\mathcal{C}$ from $1$ to $n$. Let $\mathcal{I}, \mathcal{P}^{\rm gbl}, \mathcal{P}^1, \mathcal{P}^2$ denote the sets of indices of systematic symbols, global parities, local parities introduced in step $1$ of Construction II, and local parities introduced in step $2$ of Construction II (to enable all-symbol locality), respectively. Step $1$ of Construction II involves splitting $(t-1)$ columns of the generator matrix $\bar{\Gm}_{\rm Gab}$. Let $\{\mathcal{P}^1_i\}_{i = 1}^{t-1}$ denote the sets of indices of local parities obtained from $(N + i)$-th column  of $\bar{\Gm}_{\rm Gab}$. Note that $|\mathcal{P}^1_i| =\frac{k}{r}$~for all $i \in [t-1]$ and $\mathcal{P}^1 = \bigcup_{i = 1}^{t-1}\mathcal{P}^1_{i}$. 

Next, we consider two cases for node erasure patterns:

\textit{Case 1:} There are at most $ N + \frac{N}{r} - k - \frac{k}{r} + 1$ erasures among the symbols indexed by the set $\Lc = \mathcal{I}\cup\mathcal{P}^{\rm gbl}\cup \mathcal{P}^2$. Note that the code obtained from puncturing $\mathcal{C}$ on $[n] \backslash \mathcal{L}$, {\em i.e.}, $\mathcal{C}_{\mathcal{L}}$, is a $d_{\min}$-optimal $(N + \frac{N}{r}, k)$ code with all-symbol locality $r$ \cite{SRV12, RKSV12}. These codes have minimum distance $N + \frac{N}{r} - k - \frac{k}{r} + 2$. Therefore, we can recover $k$ message symbols $\mathbf{m}$ from $\mathcal{C}_{\mathcal{L}}$ even after at most $ N + \frac{N}{r} - k - \frac{k}{r} + 1$ erasures in $\Cc_{\mathcal{L}}$.

\textit{Case 2:} There are $ N + \frac{N}{r} - k - \frac{k}{r} + 1 + x$~($1 \leq x \leq t - 1$)~erasures in $\Cc_{\mathcal{L}}$ and $t -1 - x$ erasures among the symbols indexed by the set $\mathcal{P}^1$. In this case, we obtain evaluation of a linearized polynomial $\mathpzc{f}(\cdot)$, which has $\widetilde{\mathbf{m}} =  \mathbf{m}\left(\Gm^{1}_{\rm Gab}\right)^{-1}$ as its coefficients, at $k - x$ linearly independent (over $\mathbb{F}_q$) points in $\mathbb{F}_{q^{M}}$. 

In this case, there are $t - 1 - x$ erasures among the symbols indexed by the set $\mathcal{P}^1$. In worst case, these erasure are spread in $t - 1 - x$ sets in $\{\mathcal{P}_{1}^1, \ldots, \mathcal{P}_{t-1}^{1}\}$. Let $\{\mathcal{P}^1_{i_1}, \ldots, \mathcal{P}^1_{i_{x}}\}$ be the sets corresponding to local parities that do not have any node erasure. We can combine $\frac{k}{r}$ local parities in each of these sets to obtain $x$ global parities, which correspond to evaluations of the linearized polynomial $\mathpzc{f}(\cdot)$ at $x$ linearly independent (over $\mathbb{F}_q$) points in $\mathbb{F}_{q^{M}}$. Note that these $x$ points are linearly independent from $k-x$ linearly independent (over $\F_q$) points associated with unerased symbols in $\Cc_{\mathcal{L}}$. Therefore, we get evaluations of $\mathpzc{f}(\cdot)$ at $k$ linearly independent (over $\mathbb{F}_q$) points in $\mathbb{F}_{q^{M}}$, which allows us to recover $\widetilde{\mathbf{m}}$. Given  $\widetilde{\mathbf{m}}$, we obtain $\mathbf{m}$ as  $\mathbf{m} = \widetilde{\mathbf{m}}{\Gm}^{1}_{\rm Gab}$.

This completes the proof. 
\end{proof}

\subsection{{\color{black}Explicit designs of} $\mathbf{R}$ for Constructions I and II}
\label{subsec:finding_S}

Construction I and II utilize a $k \times t\frac{k}{r}$ matrix $\Rm$ with $0$s and $1$s which satisfies specific requirements on the structure of its column and row supports (see Assumption~\ref{assump:1}). The columns of the matrix $\Rm$ are used to generate local parities for an $(n, k, r, t)$-LRC in Construction I. Similarly, for $(n, k, r, t)$-LRCs with all-symbol locality, the local parities of {\color{black}the} systematic symbols are designed according to the columns of $\Rm$.

As highlighted in Remark~\ref{rem:design_use}, a $2$-$(k, b, c,  r, 1)$ resolvable design $(\Xc, \Bc)$ with $c \geq t$ parallel classes allows us to obtain the matrix $\Rm$. 
Here, we discuss the application of a family of $2$-$(k$ = $q^{3}+1, b$ = $q^2(q^2 - q + 1), c$ = $q^2,  r$ = $q + 1, 1)$-resolvable designs \cite[Theorem 5.3.9]{IonShri}, for a prime power $q$,  to generate $(n, k, r, t)$-LRCs using Construction I and II. By scaling $q$, one can construct $(n, k, r, t)$-LRCs with $r = \Theta(k^{1/3})$, $t = r^{1-\epsilon}$ and distance $n - k - ko(1) + o(k) + 1$, {\em i.e.}, $(n, k, r, t)$-LRCs with orderwise {\color{black}the} same {\color{black}distance} of an $(n, k)$-MDS code.

By moving away from resolvable designs, we show that another construction for the matrix $\Rm$ follows from the work by Tamo et al.~\cite{zigzag13} on codes with optimal rebuilding ratio. The matrix $\Rm$ obtained from \cite{zigzag13} allows us to obtain $(n, k, r, t)$-LRCs with $r = \Theta(\frac{\log(k)}{\log\log(k)})$, $t = \Theta(r^{1-\epsilon})$ and orderwise {\color{black}the} same {\color{black}distance} as that of an $(n, k)$ MDS code.

\subsubsection{Codes with locality $\Theta(k^{\frac{1}{3}})$}

Let $q$ be a prime power. For such $q$ there exists a resolvable $2$-$(k = q^{3}+1, b = q^2(q^2 - q + 1), c = q^2,  r = q + 1, \lambda = 1)$ design \cite[Theorem 5.3.9]{IonShri}. We refer reader to \cite{IonShri} for the description of these resolvable designs. Given a resolvable design from this family, we can construct an $(n, k, r, t)$-LRC with  $k = q^{3} + 1$, $r = q + 1 = \Theta(k^{\frac{1}{3}})$ and $t = \Theta(r^{1-\epsilon})$. For such a code $\cC$ we have, 
\begin{align}
d_{\min}(\cC) &= n - k - \ceilb{\frac{kt}{r}} + t + 1 \nonumber \\
& = n - k - ko(1) + o(k) + 1. \nonumber
\end{align}
Here, we have used the fact that $\frac{t}{r} = o(1)$ and $t = o(r) = o(k)$. For an $(n,k)$ MDS code $\cC^{\rm MDS}$ we have,
$$
d_{\min}(\cC^{\rm MDS}) = n - k + 1.
$$
Therefore, 
\begin{align}
\frac{d_{\min}(\cC)}{d_{\min}(\cC^{\rm MDS})}& = \frac{n - k - \ceilb{\frac{kt}{r}} + t + 1}{n - k + 1} \nonumber \\
& = \frac{n - k - ko(1) + o(k) + 1}{n - k + 1}, \nonumber
\end{align}
which tends to $1$ as we scale both $n$ and $k$ for a fixed rate, i.e, fixed $\frac{k}{n}$. Here, we note that $n \geq k  + \frac{kt}{r} = k(1 + o(1))$. Thus, rate arbitrarily close to $1$ is possible for suitable choices of the parameters $r$ (or $q$) and $t$.
\subsubsection{Codes with locality $\Theta(\frac{\log(k)}{\log\log(k)})$}

In this subsection, we describe an approach to construct the matrix $\Rm$ with locality parameter $r = \Theta(\frac{\log(k)}{\log\log(k)})$. This particular construction follows from the work of Tamo et al. on MDS array codes with optimal rebuilding~\cite{zigzag13}. Here, we explain how the generation process for parity symbols in zigzag codes~\cite{zigzag13} implicitly constructs the matrix $\Rm$ with $k, r,$ and $t$ such that $k = rt^r$. Note that one can choose $r = \Theta(\frac{\log(k)}{\log\log(k)})$, $t = \Theta(r^{1-\epsilon})$ to satisfy $k = rt^r$.
Given a set of $k = rt^r$ elements $\Xc = \{1, 2, \ldots, k= rt^{r}\}$, we construct a collection of $r$-sized blocks $\Bc$ from elements of $\Xc$. Each block in $\Bc$ is essentially associated with a parity symbol in a zigzag code~\cite{zigzag13}. The matrix $\Rm$ is then chosen as the $k \times t^{r+1}$ indicator matrix of the collection $(\Xc, \Bc)$ (note that $\frac{k}{r}t = t^r\cdot t = t^{r+1}$). 

First, partition the set $\Xc$ into $r$ subsets $\{\Xc_j\}_{j =1}^{r}$ such that $\Xc_j = \{(j-1)t^r+1,\ldots, jt^r\}$. We index $t^r$ elements of each set $\Xc_j$ from $0$ to $t^r-1$ and denote $i$-th elements of $\Xc_j$ by $x_{i,j}$. Now consider $r$ vectors $\{e_1,\ldots, e_{r}\}$ in $\ZZ_t^r$ such that $e_j$ has $1$ at $j$-th position and zeros at other $r-1$ positions. Next, generate the collection of blocks $\Bc = \{B_j\}_{j = [t^{r+1}]} = \{Z^{l}_s\}_{l \in [1:t], s \in [0:t^r-1]}$ as follows. For each $l \in [1:t]$ and $s \in [0:t^r-1]$, take 
\begin{align}
\label{eq:Z_block}
Z^{l}_s := \{x_{i,j}: i + (l-1)e_j = s~({\rm mod}~r)\}.
\end{align}
In the definition of $Z^{l}_s$, we use $i$ to denote a unique $r$-dimensional vector associated with $i$ in $\ZZ_t^r$. Given $(\Xc, \Bc)$, we construct $k\times t^{r+1}$ $0/1$ matrix $\Rm$ as follows. 
\begin{align}
\Rm(i,j) = \begin{cases} 1, & \mbox{if } x_i \in B_{j} \\ 0, & \mbox{otherwise}. \end{cases} \nonumber
\end{align}
Note that we require the matrix $\Rm$ to satisfy two requirements (see Assumption~\ref{assump:1}). The first requirements follows from the construction of $\Bc$ as, for each $l \in [t]$,~$t^r$ blocks $\{Z^{l}_s\}_{s \in [0:t^r-1]}$ partition all the elements in $\Xc$.  Moreover, each block, say $Z^l_s$, contains exactly $r$ elements (see \eqref{eq:Z_block}). It remains to show that the second requirement for $\Rm$ also holds. This follows from the next claim that any pair of elements of $\Xc$ appears in at most one block in $\Bc$. 

\begin{claim}
Each pair of elements $\{x_{i,j}, x_{i',j'}\} \in \Xc$ is present in at most one block in $\Bc$.
\end{claim}

\begin{proof}
Let us assume that the opposite is true and there exists a pair of distinct elements $\{x_{i,j}, x_{i',j'}\} \in \Xc$ which appear in two blocks $Z^{l}_s$ and $Z^{l'}_{s'}$. It is easy to check from the construction of sets $ \{Z^{l}_s\}_{l \in [1:t], s \in [0:t^r-1]}$ that $j \neq j'$ as only one element from a set $\Xc_j$ participate in a block. Similarly, we have $l \neq l'$ as for a fixed $l \in [t]$, the blocks in $\{Z^{l}_s\}_{s \in [0:t^r-1]}$ partition the set $\Xc$; as a result, no element can appear in two blocks in $\{Z^{l}_s\}_{s \in [0:t^r-1]}$.

Since $\{x_{i,j}, x_{i',j'}\} \in Z^{l}_{s}$ and $\{x_{i,j}, x_{i',j'}\} \in Z^{l'}_{s'}$, we have,
\begin{align}
\label{eq:zigzag1}
i + (l -1)e_j = s~({\rm mod}~r), \\
i + (l' -1)e_j = s'~({\rm mod}~r) \label{eq:zigzag2}
\end{align}
and 
\begin{align}
\label{eq:zigzag3}
i' + (l -1)e_{j'} = s~({\rm mod}~r), \\ 
i' + (l' -1)e_{j'} = s'~({\rm mod}~r).  \label{eq:zigzag4}
\end{align}

Note that we use vector representation of $i, i', s,$ and $s'$ in the above equations. Subtracting \eqref{eq:zigzag2} from \eqref{eq:zigzag1}, and \eqref{eq:zigzag4} from \eqref{eq:zigzag3}, we obtain
\begin{align}
\label{eq:zigzag5}
 (l - l')e_{j} = s - s'~({\rm mod}~r),\\ 
 (l -l')e_{j'} = s -s'~({\rm mod}~r).  \label{eq:zigzag6}
\end{align}
However, it is not possible to satisfy both \eqref{eq:zigzag5} and \eqref{eq:zigzag6} simultaneously as that leads to a contradiction $j = j'$. This completes the proof.
\end{proof}
Here, we note that $(\Xc, \Bc)$ obtained in this subsection is not a $2$-design as it does not satisfy the requirement (ii) in Definition~\ref{def:2-design}, {\em i.e.}, every pair $(x, y) \subset \Xc$ is present in {\em exactly} $1$ block (subset) in $\Bc$.  .


\section{Conclusions and Open Problems}

There are several important questions that remain open. It is unclear if explicit codes can be constructed to attain the general bound in Theorem~\ref{prop:general_bound}. Further, it is very interesting to explore if the trade-off between distance, locality and availability in Theorem~\ref{prop:bound_restricted} remains true for non-linear codes. 

Some combinatorial questions also arise. It seems that resolvable design requirements are not entirely necessary but it is not clear if significantly better parameters can be obtained by other methods. In addition, there are several choices of parameters $(n,k,r,t)$ where it is not clear if codes with high availability exist.

Another open question is that of multiple parallel reads of different symbols, as explored in \cite{batchcodes}. Our current work ensures that each symbol can be read by multiple processes but no guarantees are given for reading two different symbols concurrently.  It would be interesting to obtain bounds on the number of arbitrary parallel reads that can be supported by any code of a given locality and distance. 

From a practical {\color{black}point of view}, we believe that the property of multiple parallel reads will be useful for distributed storage systems with hot data. The benefits need to be quantified, possibly through queuing theoretic models or through system measurements. Regardless of current technological impact, we believe that the concept of availability is interesting and gives a fruitful direction for coding theory research.


%
\bibliographystyle{unsrt}
\bibliography{availability_ISIT}

\appendices

\section{Proof of Theorem~\ref{prop:general_bound}}
\label{appen:general_bound}
We utilize the proof technique introduced by Forbes and Yekhanin in \cite{forbes} to obtain upper bound on minimum distance of a non-linear systematic code with locality $r$.  Note that Cadambe and Mazumdar also introduce the similar proof technique to obtain field size dependent upper bound on the minimum distance in \cite{cadambe_mazumdar}. However, we closely follow the approach of Forbes and Yekhanin~\cite{forbes} in the following.

\begin{figure}[htbp]
\algrule[1pt]
\textbf{Algorithm:} Construction of sub-code $\cC' \subset \cC$.
\algrule[1pt]
\begin{algorithmic}[1]
\REQUIRE $n, k, r, t$ and $(n, k)$ systematic code $\cC \subseteq \Sigma^{n}$ with $(r, t)$-availability ($\Sigma = \F_q$).
\STATE $\cC_0 = \cC$
\STATE $j = 0$  
\WHILE{$|\cC_{j}| > q$}
\STATE $j = j + 1$.
\STATE Choose $i_j$ such that $i_j \notin \Rc_{j - 1} := \bigcup_{j' \in [j-1]}\left(\Sc_{j'} \cup \{i_{j'}\}\right)$.
\STATE Let $\Sc_{j} = \Sc^{1}_{i_j} \cup \cdots \cup \Sc^{t}_{i_j}$ be union of $t$ disjoint local groups of $i_j$.
\STATE Let $\vec{\sigma}_{\,j} \in \Sigma^{|\Sc_{j}|}$ be the most frequent element in the multi-set $\{\vec{x}|_{\Sc_j}: \vec{x} \in \cC_{j-1}\}$.
\STATE Define $\cC_{j} := \{\vec{x}: \vec{x} \in \cC_{j-1}, \vec{x}|_{\Sc_j} = \vec{\sigma}_{\,j}\}$.
\IF{$1 < |\cC_{j}| \leq  q$}
\STATE $\cC' = \cC_{j}$.
\STATE \textbf{end while}
\ELSIF{$|\cC_{j}| = 1$}
\STATE Pick a maximal subset $\widetilde{\Sc}_j \subseteq \Sc_j$ such that $|\tilde{\cC_j}| > 1$, where $\widetilde{\cC}_{j} := \{\vec{x}: \vec{x} \in \cC_{j-1}, \vec{x}|_{\widetilde{\Sc}_j} = \vec{\gamma}_{\,j}\}$ and $\vec{\gamma}_{\,j} \in \Sigma^{|\tilde{\Sc}_{j}|}$ be the most frequent element in the multi-set $\{\vec{x}|_{\tilde{\Sc}_j}: \vec{x} \in \cC_{j-1}\}$.
\STATE $\cC' = \widetilde{\cC}_j$. 
\STATE \textbf{end while}.
\ENDIF
\ENDWHILE
\ENSURE $\cC'$.
\end{algorithmic}
\algrule[1pt]
\caption{Construction of sub-code $\cC' \subset \cC$.}
\label{algo:subcode}
\end{figure}

\begin{proof}
Given a systematic $(n, k, r, t)$-LRC, we construct a subcode $\cC' \subset \cC$ such that most of the coordinates of codewords in $\cC$ are fixed. Now, puncturing the codewords of $\cC'$ on these fixed coordinates provides us with a new code $\cC''$ which has the same dimension as $\cC'$ but at the same time has codewords of significantly smaller length as compare to $n$. Note that $d_{\min}(\cC'') = d_{\min}(\cC') \geq d_{\min}(\cC)$. We then apply the Singleton bound on $\cC''$ to obtain an upper bound on $d_{\min}(\cC'')$, which subsequently gives us an upper bound on $d_{\min}(\cC)$. An algorithm to construct the subcode $\cC' \subset \cC$ is presented in Fig.~\ref{algo:subcode}.  

Note that the algorithm in Fig.~\ref{algo:subcode} is well defined in the sense that it is always possible to find $i_j$ at line $5$. Since the algorithm reaches at line $5$ only if $\cC_{j-1} > 1$, there exists two distinct codewords say $\cv_1$ and $\cv_2$ in $\cC_{j-1}$. Note that both $\cv_1$ and $\cv_2$ are identical at coordinates specified by $\Rc_{j-1}$ as coordinates at $\{\Sc_{j'}\}_{j' < j}$ and consequently coordinates at $\{i_{j'}\}_{j' < j}$ have been fixed in step $8$ of previous iterations. Therefore, they have to differ at a coordinate which is not in $\Rc_{j-1}$. Moreover, none of its disjoint $t$ repair groups $\Sc^1_{i_j},\ldots, \Sc^{t}_{i_j}$ are completely contained in $\Rc_{j-1}$; otherwise, the $i_j$-th coordinate would have been fixed as coordinate $i_j$ is fixed once all coordinates in any one of its repair groups are fixed. 

Before we proceed with analysis, we define $\Ac_{j} = \Sc_{j} \backslash \Rc_{j - 1}$ and $a_j = |\Ac_{j}|$. Assuming that the while loop in Fig.~\ref{algo:subcode} ends with $j = \ell$, for $j \in [\ell]$, we have 
$$
\Rc_{j} = \bigsqcup_{j' \in [j]}\left(\Ac_{j'} \sqcup \{i_{j'}\}\right).
$$
At line $7$, taking into account locality due to $t$ disjoint repair groups, there are at most $q^{a_j - (t-1)}$ possibilities for $\vec{\sigma}_{\,j}$; thus, we have $|\cC_j| \geq |\cC_{j-1}|/q^{a_j - (t-1)}$. Note that there are two possibilities for last iteration $j = \ell$. The sub-code $\cC'$ can be obtained at line $10$ or at line $14$. In the following, for the ease of exposition, we assume that the $\cC'$ is obtained at line $10$ of $\ell$-th iteration. Other case can be analyzed using ideas similar to those employed in \cite{PapDim12, RKSV12, WangZhang}.

Since the construction algorithm for $\cC'$ ends with $j = \ell$, we have $|\cC_{\ell}| \leq  q$, or
\begin{align}
1 \geq \log_{q}|\cC_{\ell+1}|&\geq k - \sum_{j = 1}^{\ell}\left(a_j - (t-1)\right).
\end{align}
Now, using that $a_j \leq |\Sc_j| \leq tr$, we get 
\begin{align}
k - 1 \leq \ell(tr - t + 1), 
\end{align}
which gives us that 
\begin{align}
\label{eq:no_iter}
\ell \geq \ceilb{\frac{k- 1}{tr - t + 1}}.
\end{align}
Note that sub-code $\cC' = \cC_{\ell}$. Therefore, 
\begin{align}
\label{eq:new_dim}
\log_q|\cC'| & = \log_q|\cC_{\ell}| \nonumber \\
&\geq \log_q|\cC| - \sum_{j = 1}^{\ell}\left(a_j - (t-1)\right) \nonumber \\
& = k - \sum_{j = 1}^{\ell}a_j + \ell(t-1) \nonumber \\
&\overset{(a)}{=} k - |\Rc_{\ell}| + \ell + \ell(t-1) \nonumber \\
& = k - |\Rc_{\ell}| + t\ell,
\end{align}
where $(a)$ follows from the fact that $|\Rc_{\ell}| = |\bigsqcup_{j' \in [\ell]}\left(\Ac_{j'} \sqcup \{i_{j'}\}\right)| = \sum_{i = 1}^{\ell}a_j + \ell$. Now, we define $\cC'' = \cC'|_{\Rc_{\ell}}$ which denotes the sub-code obtained by puncturing $\cC'$ on indices denoted by $\Rc_{\ell}$. Since all codewords in $\cC'$ are fixed for all coordinates indexed by $\Rc_{\ell}$, we have $|\cC''| = |\cC'|$ and $d_{\min}(\cC'') = d_{\min}(\cC')$. Moreover, the length of the codewords in $\cC''$ is $n - |\Rc_{\ell}|$. Next, applying the Singleton bound on $\cC''$ gives us
\begin{align}
\label{eq:dmin_mid}
d_{\min}(\cC'') &\leq (n - |\Rc_{\ell}|) - \log_q|\cC''| + 1 \nonumber \\
&\overset{(b)}{\leq} n - |\Rc_{\ell}| - (k - |\Rc_{\ell}| +t\ell) + 1\nonumber \\
& = n - k - t\ell + 1,
\end{align}
where $(b)$ follows from \eqref{eq:new_dim} and the fact that $|\cC''| = |\cC'|$. Now, combining \eqref{eq:dmin_mid} and \eqref{eq:no_iter} gives us
\begin{align}
d_{min}(\cC'') \leq n - k +1 - t\ceilb{\frac{k- 1}{tr - t + 1}}.
\end{align}
Using the fact that $d_{\min}(\cC) \leq d_{\min}(\cC'') = d_{\min}(\cC')$, we obtain
\begin{align}
d_{\min}(\cC) \leq n - k +1 - t\ceilb{\frac{k- 1}{tr - t + 1}}.
\end{align}
Next, we can use $t\ceilb{\frac{k- 1}{tr - t + 1}} \geq \ceilb{\frac{kt- t + 1}{tr - t + 1}} - 1$~\cite{WangZhang} to claim that
\begin{align}
d_{\min}(\cC) \leq n - k + 1 - \left(\ceilb{\frac{kt - t + 1}{tr - t + 1}} - 1 \right).
\end{align}
This completes the proof.
\end{proof}

\end{document}